\documentclass[
final
]{dmtcs-episciences}

\usepackage[utf8]{inputenc}
\usepackage{subfigure}
\usepackage{graphicx, amsmath, amsthm, amssymb, tikz}

\usepackage[round]{natbib}

\author[Alison Hsiang-Hsuan Liu and Nikhil S.~Mande]{Alison Hsiang-Hsuan Liu\affiliationmark{1}
  \and Nikhil S.~Mande\affiliationmark{2}\thanks{This work has been supported by Research England funding to enhance research culture, and the Royal Society International Exchanges grant ``Power of Knowledge in Explorable Uncertainty''.}
}
\title[Instance complexity of Boolean functions]{Instance complexity of Boolean functions}
\affiliation{
  Utrecht University, Utrecht, The Netherlands.\\
  University of Liverpool, Liverpool, UK.}
\keywords{Instance complexity, Boolean functions, query complexity}
\begin{document}
\publicationdata{vol. 28:2}{2026}{27}{10.46298/dmtcs.15917}{2025-06-23; 2025-06-23; 2026-04-13}{2026-05-06}

\newcommand{\zone}{\{0,1\}}
\newcommand{\T}{\mathcal{T}}
\newcommand{\DT}{\mathsf{DT}}
\newcommand{\C}{\mathsf{C}}
\newcommand{\Cmin}{\mathsf{C}_{\mathsf{min}}}
\newcommand{\IC}{\mathsf{InstC}}
\newcommand{\XOR}{\mathsf{XOR}}
\newcommand{\IND}{\mathsf{IND}}
\newcommand{\bin}{\mathsf{bin}}
\newcommand{\MAJ}{\mathsf{MAJ}}
\newcommand{\AND}{\mathsf{AND}}
\newcommand{\OR}{\mathsf{OR}}
\newcommand{\GT}{\mathsf{GT}}
\newcommand{\OMB}{\mathsf{OMB}}
\newcommand{\CONN}{\mathsf{CONN}}
\newcommand{\kCL}{\mathsf{CL}_{k}}
\newcommand{\cA}{\mathcal{A}}
\newcommand{\cB}{\mathcal{B}}

\newcommand{\cbra}[1]{\left\{#1\right\}}

\newtheorem{theorem}{Theorem}[section]
\newtheorem{corollary}[theorem]{Corollary}
\newtheorem{remark}[theorem]{Remark}
\newtheorem{lemma}[theorem]{Lemma}
\newtheorem{claim}[theorem]{Claim}
\newtheorem{question}[theorem]{Question}
\newtheorem{conjecture}[theorem]{Conjecture}
\newtheorem{observation}[theorem]{Observation}
\newtheorem{definition}[theorem]{Definition}

\maketitle
\begin{abstract}
    In the area of query complexity of Boolean functions, the most widely studied cost measure of an algorithm is the worst-case number of queries made by it on an input. Motivated by the most natural cost measure studied in online algorithms, the \emph{competitive ratio}, we consider a different cost measure for query algorithms for Boolean functions that captures the ratio of the cost of the algorithm and the cost of an optimal algorithm that knows the input in advance. The cost of an algorithm is its largest cost over all inputs. Grossman, Komargodski and Naor [ITCS'20] introduced this measure for Boolean functions, and dubbed it \emph{instance complexity}. Grossman et al.~showed, among other results, that monotone Boolean functions with instance complexity 1 are precisely those that depend on no variables or one variable.

    We complement the above-mentioned result of Grossman et al.~by completely characterizing the instance complexity of \emph{symmetric} Boolean functions. As a corollary we conclude that the only symmetric Boolean functions with instance complexity 1 are the Parity function and its complement. We also study the instance complexity of some graph properties like Connectivity and $k$-clique containment.

    In all the Boolean functions we study above, and those studied by Grossman et al., the instance complexity turns out to be the ratio of query complexity to \emph{minimum certificate complexity}. It is a natural question to ask if this is the correct bound for \emph{all} Boolean functions. We show a negative answer in a very strong sense, by analyzing the instance complexity of the Greater-Than and Odd-Max-Bit functions. We show that the above-mentioned ratio is linear in the input size for both of these functions, while we exhibit algorithms for which the instance complexity is a constant.
\end{abstract}

\section{Introduction}
In the typical setting of online algorithms, an algorithm designer's task is to design an efficient algorithm that is geared towards receiving inputs in an online fashion. More specifically, the input is revealed to an (online) algorithm piece by piece. On each revelation of a piece of the input, the algorithm needs to make \emph{irrevocable} decisions. A natural cost measure of an online algorithm is the \emph{competitive ratio}, which is defined as the biggest ratio of the algorithm's cost to the optimal offline algorithm's cost on the same input, where the optimal offline algorithm knows the whole input. 

Worst-case analysis is a setting studied in various different models, for instance the query complexity model, which is relevant to our discussion. Unlike the worst-case setting, competitive analysis does not only focus on measuring the performance of algorithms on single ``hard" inputs but is more representative of the performance of algorithms on all inputs as a whole. Moreover, it reveals how uncertainty affects the quality of decisions. The measure of competitive ratio has gained interest of late in the context of \emph{explorable uncertainty}. 
In this model, instead of completely unknown inputs, an algorithm receives an uncertain input with the promise that every numerical value sits in an interval, where the realization of the value can be learned by \emph{exploration}. 
The concept of explorable uncertainty has raised a lot of attention and has been studied on different problems, such as Pandora's box problem~\cite{DBLP:conf/soda/DingFHTX23}, sorting~\cite{DBLP:journals/tcs/HalldorssonL21}, finding the median~\cite{DBLP:conf/stoc/FederMPOW00}, identifying a set with the minimum-weight among a given collection of feasible sets~\cite{DBLP:journals/tcs/Erlebach0K16}, finding shortest paths~\cite{DBLP:journals/jal/FederMOOP07}, computing minimum spanning trees~\cite{DBLP:conf/stacs/HoffmannEKMR08}, etc. In most of these problems, the cost of an algorithm on an input is naturally considered to be the competitive ratio, which is the ratio of the number of explorations made (i.e., the number of \emph{queries} made to the input) to the number of explorations made by the best offline algorithm that knows the input in advance.

We consider query algorithms for Boolean functions, which are functions mapping an $n$-bit input to a single bit. A query algorithm for a Boolean function is represented by a \emph{decision tree}, which is a rooted binary tree with internal nodes labeled by variables, edges labeled by values in $\zone$, and leaves labeled by values in $\zone$. The decision tree evaluates an input in the natural way beginning at the root, and traversing a path until a leaf, at which point it outputs the value at the leaf. In the usual query complexity setting, the cost of this decision tree is its depth. The query complexity of a Boolean function $f$, also known as decision tree complexity of $f$, is the number of queries an optimal query algorithm makes on a worst-case input, i.e., the minimum depth of a decision tree computing $f$. 

It is natural to consider the cost measure described in the second paragraph for query algorithms for Boolean functions rather than the more general class of functions studied in the online algorithms setting. In the setting of explorable uncertainty, this corresponds to the input being unknown in the beginning, and each bit has a ``uncertainty range'' of $\zone$ rather than an ``uncertainty interval'' like in optimization problems as from the first two paragraphs. To this end, \cite{GKN20} only recently initiated the study of \emph{instance complexity} of Boolean functions, where instance complexity of an algorithm is the maximum over all inputs of the ratio of the number of queries made on the input to the number of queries made by an optimal algorithm that knows the whole input. We define the instance complexity of a Boolean function to be the minimum instance complexity of an algorithm that solves it. Intuitively, studying the instance complexity of Boolean functions overcomes some drawbacks that the usual query complexity model has. For example, the query complexity of a function could be very large owing to the hardness of one single input, but it may be the case that even algorithms with prior knowledge about this input may require lots of queries to certify the function's evaluation on it. Nevertheless, it does turn out to be single ``hard" inputs that contribute to the large instance complexity of the AND and OR functions, just as in the usual query complexity setting, but we show that this is not always the case.

\subsection{Our contributions}
We continue the study of instance complexity of Boolean functions initiated by Grossman et al. Among other results, they characterized monotone Boolean functions which are strictly instance optimizable, i.e., monotone Boolean functions which have instance complexity equal to 1. We complement this result by completely characterizing the instance complexity of \emph{symmetric} Boolean functions in terms of their univariate predicates. We refer the reader to Section~\ref{sec:prelims} for formal definitions.

For a symmetric function $f : \zone^n \to \zone$, let the integers $0 \leq \ell_0(f) \leq \ell_1(f) \leq n$ denote the end points of the largest interval of Hamming weights in which $f$ is a constant.
\begin{theorem}\label{thm:mainsymintro}
    Let $f : \zone^n \to \zone$ be a symmetric Boolean function. Then,
    \[
    \IC(f) = \frac{n}{\ell_0(f) + n - \ell_1(f)}.
    \]
\end{theorem}

In particular it follows that the only symmetric Boolean functions that are strictly instance optimizable are the Parity function and its complement. In the process we show that the instance complexity of a symmetric Boolean function $f$ is the ratio between its query complexity $\DT(f)$ (which is the decision tree complexity of $f$, equaling the number of input variables for symmetric $f$) and its \emph{minimum certificate complexity} $\Cmin(f)$, which is the smallest number of variables one needs to fix in order to fix the function value to a constant. In other words, $\Cmin(f)$ equals the minimum co-dimension of an affine subcube on which $f$ is a constant.

More generally, one can observe that $\DT(f)/\Cmin(f)$ is an upper bound on the instance complexity of any Boolean function $f$. Along with showing that this bound is attained for all symmetric functions, we also show that this bound is attained for some \emph{graph properties} like Connectivity and $k$-clique containment. 
Let $\CONN$ and $\kCL$ denote the Connectivity and $k$-Clique problems, respectively (See Definitions~\ref{defn:conn} and~\ref{defn:kcl}).
\begin{theorem}\label{thm:graphintro}
    Let $n$ be a sufficiently large positive integer and $k$ be a positive integer satisfying $k^3 \leq n^2/4$. Then,
    \[
    \IC(\CONN) = \frac{\DT(\CONN)}{\Cmin(\CONN)} = \frac{\binom{n}{2}}{n-1}, \qquad \IC(\kCL) = \frac{\DT(\kCL)}{\Cmin(\kCL)} = \frac{\binom{n}{2}}{\binom{k}{2}}.
    \]
\end{theorem}

In view of this one may expect that the instance complexity of \emph{all} Boolean functions $f$ equals the ratio $\DT(f)/\Cmin(f)$. We show that this is false in a very strong way using two examples: the Greater-Than function on $2n$ input variables, and the Odd-Max-Bit function on $n$ variables, denoted by $\GT_n$ and $\OMB_n$, respectively (see Definitions~\ref{defn:gt} and~\ref{defn:omb}). Both of these functions have $\DT(f)/\Cmin(f) = \Theta(n)$ but instance complexity $O(1)$.
\begin{theorem}\label{thm:gtombresultsintro}
    For all positive integers $n$, we have
    \begin{align*}
        \DT(\GT_n) & = 2n, \qquad \Cmin(\GT_n) = 2, \qquad \IC(\GT_n) \leq 2\\
        \DT(\OMB_n) & = n, \qquad \Cmin(\OMB_n) = 1, \qquad \IC(\OMB_n) \leq 2.
    \end{align*}
\end{theorem}

While none of our results are deep or technically involved, our main goal is to bring to light the natural and interesting complexity measure of \emph{instance complexity} of Boolean functions.
Some interesting open questions that remain are to characterize the instance complexity of monotone or linear threshold functions in terms of some combinatorial parameter, just as we were able to do for symmetric functions.

\section{Preliminaries}\label{sec:prelims}
For a positive integer $n$, we use the notation $[n]$ to denote the set $\cbra{1, 2, \dots, n}$. For a string $x \in \zone^n$ and a set $S \subseteq [n]$, we denote by $x_S$ the string in $\zone^S$ that is the restriction of $x$ to the coordinates indexed by $S$. 
For a string $x \in \zone^n$, let $|x|$ denote the Hamming weight of $x$, that is, the number of 1s in $x$.
Let $\XOR_n : \zone^n \to \zone$ denote the Parity function on $n$ input bits, that outputs 1 iff the number of 1's in the input is odd. Let $\MAJ_n : \zone^n \to \zone$ denote the Majority function that outputs 1 iff the number of 1's is at least the number of 0's in the input.
Define the Indexing function as follows.
\begin{definition}[Indexing Function]
For a positive integer $m$, define the Indexing function, denoted $\IND_m : \zone^{m + 2^m} \to \zone$, by
\[
\IND_m(x, y) = y_{\bin(x)},
\]
where $\bin(x)$ denotes the integer in $[2^m]$ represented by the binary expansion $x$.
\end{definition}
In the above definition, we refer to $\cbra{x_i : i \in [m]}$ as the \emph{addressing variables}, and $\cbra{y_j : j \in [2^m]}$ as the \emph{target variables}.

A deterministic decision tree is a rooted binary tree. Internal nodes are labeled by variables $x_i$, and leaves are labeled by values in $\zone$. Given an input $x \in \zone^n$, the tree's evaluation on the input proceeds in the natural way: traverse the relevant edge depending on the value of the variable of the node until reaching a leaf, at which point the value at the leaf is output.
A decision tree $\T$ is said to compute a Boolean function $f : \zone^n \to \zone$ if  the value output by $\T$ on input $x$ equals $f(x)$ for all $x \in \zone^n$. Denote the number of queries made by $\T$ on input $x$ by $\T(x)$. The decision tree complexity of $\T$ is the worst-case number of queries made, i.e., its depth.

The decision tree complexity (also called \emph{deterministic query complexity}) of $f$, denoted $\DT(f)$, is defined as follows.
\[
\DT(f) := \min_{\T:\T~\textnormal{is a DT computing}~f}
\textnormal{depth}(\T).
\]
\emph{Certificate complexity} captures non-deterministic query complexity. A certificate for an input $x \in \zone^n$ to a function $f : \zone^n \to \zone$ is a set $S \subseteq [n]$ such that $f(y) = f(x)$ for all $y \in \zone^n$ with $y_S = x_S$. The certificate complexity of $f$ at input $x$, denoted $\C(f, x)$, is the minimum size of such a set $S$. The certificate complexity of $f$, denoted $\C(f)$, is defined as follows.
\[
\C(f) = \max_{x \in \zone^n}\C(f, x).
\]
We define another complexity measure that has been studied in the past: \emph{minimum certificate complexity}. This is the minimum co-dimension of an affine subcube on which the underlying function is a constant.
\begin{definition}
    For a Boolean function $f : \zone^n \to \zone$, define the \emph{minimum certificate complexity} of $f$, denoted $\Cmin(f)$, to be
    \[
    \Cmin(f) := \min_{x \in \zone^n} \C(f, x).
    \]
\end{definition}
The interested reader may refer to the survey~\cite{BW02} for an introduction to query complexity and related measures of Boolean functions.

\cite{GKN20} introduced the measure of \emph{instance complexity} of a Boolean function. Although they do not frame it as below, we find the form below convenient as it cleanly captures a complexity measure.

\begin{definition}\label{defn:instancecomp}
    For a Boolean function $f : \zone^n \to \zone$, an input $x \in \zone$ and a decision tree $\T$ that computes $f$, define the \emph{instance complexity of $f$ at input $x$ w.r.t.~$\T$}, which we denote by $\IC(f, x, \T)$, to be
    \[
    \IC(f, x, \T) := \frac{\T(x)}{\C(f, x)}.
    \]
    Define the \emph{instance complexity of $f$ w.r.t.~$T$} to be
    \[
    \IC(f, \T) := \max_{x \in \zone^n}\IC(f, x, \T).
    \]
    Finally, define the \emph{instance complexity of $f$}, which we denote by $\IC(f)$, to be
    \[
    \IC(f) = \min_{\T : \T~\mathrm{computes}~f}\IC(f, \T).
    \]
\end{definition}
In other words, the instance complexity of a function $f$ is small if there exists a decision tree solving it such that for \emph{all} inputs, the cost of the decision tree is not much larger that the cost of an optimal decision tree on that input (i.e., the certificate complexity of $f$ at that input).
Functions of instance complexity 1 are precisely those that Grossman et al.~refer to as \emph{strictly $D$-instance optimizable}.

\subsection{Related work}
\cite{GKN20} showed the following, among other results, regarding the instance complexity of specific Boolean functions.

\begin{lemma}[{\cite[Section 3]{GKN20}}]\label{lem:gkn}
For all positive integers $n, m$,
\begin{align*}
\IC(\XOR_n) = \IC(\IND_m) & = 1,\\
\IC(\MAJ_n) & \approx 2,\\
\IC(\AND_n) = \IC(\OR_n) & = n.
\end{align*}
\end{lemma}

The notion of instance complexity and instance optimality of functions and algorithms have been studied from various angles. See, for example,~\cite{BGN23} and the references therein.

A Boolean function $f : \zone^n \to \zone$ is said to be \emph{monotone} if $x \leq y \implies f(x) \leq f(y)$. Here $x \leq y$ represents coordinate-wise inequality. In other words, $f$ is monotone if flipping a 1 to a 0 in any input can never change the function value from 0 to 1. Examples of monotone functions are $\AND, \OR, \MAJ$. \cite[Lemma~3.4]{GKN20} characterized monotone functions $f$ that satisfy $\IC(f) = 1$. They showed that the monotone functions satisfying $\IC(f) = 1$ are precisely those that depend on either 0 or 1 variable.

A natural upper bound on $\IC(f)$ is $\DT(f)/\Cmin(f)$. This is simply because an optimal algorithm (w.r.t.~query complexity) for $f$ witnesses this: the largest possible numerator and smallest possible denominator in the first expression of Definition~\ref{defn:instancecomp} are $\DT(f)$ (witnessed by an optimal decision tree algorithm for $f$) and $\Cmin(f)$, respectively. We show in the next section that this is tight for the class of symmetric Boolean functions $f$. In the following section we analyze the instance complexity of some graph properties, and show that this bound is tight in these cases as well. In the next section we show that such an equality does not hold true for general Boolean $f$.

\begin{remark}
    Given that certificate complexity is a more well-studied measure than \emph{minimum} certificate complexity, one might ask about the relationship between $\IC(f)$ and $\DT(f)/\C(f)$. i.e., can instance complexity be characterized by the ratio of the (worst-case) query complexity of $f$ to the (worst-case) certificate complexity of $f$? For $\AND, \OR$, the former quantity, $\IC(f)$, is $n$ (by Lemma~\ref{lem:gkn}) and the latter quantity, $\DT(f)/\C(f)$, is $1$. For $\MAJ_n$, both quantities are roughly 2. It would be interesting to find an example (or show that such an example does not exist) where $\IC(f) \ll \DT(f)/\C(f)$.
\end{remark}

\section{Instance complexity of symmetric Boolean functions}\label{sec:symm}
In this section we completely characterize instance complexity of symmetric Boolean functions. This is a generalization of~\cite[Examples~3.2,~3.3]{GKN20}. 
We first formally define symmetric functions below. For a positive integer $n$, we use $S_n$ to denote the group of permutations of $n$ elements. 

\begin{definition}[Symmetric functions]
A function $f : \zone^n \to \zone$ is symmetric if for all $\sigma \in S_n$ and for all $x \in \zone^n$ we have $f(x) = f(\sigma(x))$.
\end{definition}
Equivalently, a function is symmetric iff its output only depends on the Hamming weight of the input. Hence, we may identify a symmetric function $f$ with its associated \emph{predicate}, denoted $D_f$, and defined as
\[
D_f(i) = b \qquad \textnormal{if}~|x| = i \implies f(x) = b.
\]
For a symmetric function $f : \zone^n \to \zone$, let the integers $0 \leq \ell_0(f) \leq \ell_1(f) \leq n$ denote the end points of a largest interval of Hamming weights in which $f$ is a constant.\footnote{Here, ties are broken arbitrarily if the largest such interval is not unique. Note that the results that use these measures (Claim~\ref{claim:symmcmin} and Theorem~\ref{thm:symmic}) are independent of the choice of largest interval since $\ell_0 - \ell_1$ is invariant under this choice.}
See Figure~\ref{fig:symmpred} for a pictorial description.

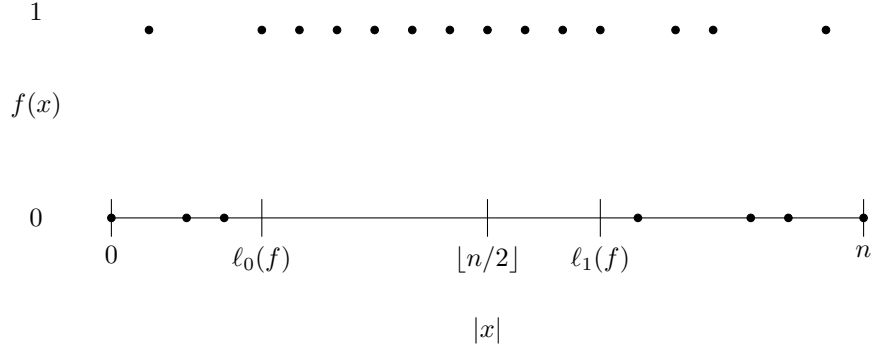
\begin{figure}
\begin{center}
    \begin{tikzpicture}[scale=0.5]
    \draw (-10,0)-- (10,0); 
    \node (yaxislabel) at (-12,3) {$f(x)$};
    \node (xaxislabel) at (0, -3) {$|x|$};
    \node (p0) at (-12,0) {$0$};
    \node (p1) at (-12,5.5) {$1$};
    \draw (0,.5) -- (0,-.5) node[below] {$\lfloor n/2\rfloor$};
    \draw (-10,.5) -- (-10,-.5) node[below] {$0$};
    \draw (10,.5) -- (10,-.5) node[below] {$n$};
    \draw (-6,.5) -- (-6,-.5) node[below] {$\ell_0(f)$};
    \draw (3,.5) -- (3,-.5) node[below] {$\ell_1(f)$};
    \draw[fill=black] (-1,5) circle (0.1);
    \draw[fill=black] (-2,5) circle (0.1);
    \draw[fill=black] (-3,5) circle (0.1);
    \draw[fill=black] (-4,5) circle (0.1);
    \draw[fill=black] (-5,5) circle (0.1);
    \draw[fill=black] (-6,5) circle (0.1);
    \draw[fill=black] (-7,0) circle (0.1);
    \draw[fill=black] (-8,0) circle (0.1);
    \draw[fill=black] (-9,5) circle (0.1);
    \draw[fill=black] (-10,0) circle (0.1);
    \draw[fill=black] (1,5) circle (0.1);
    \draw[fill=black] (2,5) circle (0.1);
    \draw[fill=black] (3,5) circle (0.1);
    \draw[fill=black] (4,0) circle (0.1);
    \draw[fill=black] (5,5) circle (0.1);
    \draw[fill=black] (6,5) circle (0.1);
    \draw[fill=black] (7,0) circle (0.1);
    \draw[fill=black] (8,0) circle (0.1);
    \draw[fill=black] (9,5) circle (0.1);
    \draw[fill=black] (10,0) circle (0.1);
    \draw[fill=black] (0,5) circle (0.1);
\end{tikzpicture}
\end{center}
\caption{Visual representation of the predicate $D_f$ of a symmetric Boolean function $f$. The interval $[\ell_0(f), \ell_1(f)]$ is the largest interval on which $f$ is constant.}
\label{fig:symmpred}
\end{figure}

We observe below that the minimum certificate complexity of a symmetric function $f$ equals $\ell_0(f) + n - \ell_1(f)$.
\begin{claim}\label{claim:symmcmin}
Let $f : \zone^n \to \zone$ be a symmetric Boolean function. Then
\[
\Cmin(f) = \ell_0(f) + n - \ell_1(f).
\]
\end{claim}
\begin{proof}
We prove the upper bound and lower bound separately.
\begin{itemize}
    \item For the upper bound, consider an input $x$ with Hamming weight $k \in [\ell_0(f), \ell_1(f)]$. Consider a set of indices $S_0$ of $\ell_0(f)$ many 1's and $S_1$ of $n - \ell_1(f)$ many 0's of it. Clearly such a set exists since $k \in [\ell_0(f), \ell_1(f)]$. Any input consistent with these input bits must have Hamming weight in $[\ell_0(f), \ell_1(f)]$, and thus must be a constant by assumption. Thus, $S_0 \cup S_1$ is a certificate for $x$ of size $\ell_0(f) + n - \ell_1(f)$.
    \item Towards the lower bound, assume towards a contradiction that there is an input $x$ with a certificate $C$ of size $w < \ell_0(f) + n - \ell_1(f)$. Suppose $C$ contains $m_0$ 0-indices and $m_1$ 1-indices, where $m_0 + m_1 = w$. The fact that $C$ is a certificate implies that $f$ must output the same value on all inputs of Hamming weight in $[m_1, n - m_0]$. The length of this interval is 
    \[
    n - m_0 - m_1 = n - w > \ell_1(f) - \ell_0(f).
    \]
    This contradicts the maximality of the interval $[\ell_0(f), \ell_1(f)]$, yielding a contradiction.\hfill\qed
\end{itemize}
\let\qed\relax
\end{proof}
As an immediate corollary, we obtain the following.
\begin{corollary}\label{cor:symmmaxcmin}
    The only symmetric functions $f : \zone^n \to \zone$ with $\Cmin(f) = n$ are the parity function $\XOR_n$ and its negation.
\end{corollary}

Our characterization of the instance complexity of symmetric Boolean functions is as follows.
\begin{theorem}[Restatement of Theorem~\ref{thm:mainsymintro}]\label{thm:symmic}
    Let $f : \zone^n \to \zone$ be a symmetric Boolean function. Then,
    \[
    \IC(f) = \frac{n}{\Cmin(f)} = \frac{n}{\ell_0(f) + n - \ell_1(f)}.
    \]
\end{theorem}
\begin{proof}
    The last equality follows from Claim~\ref{claim:symmcmin}.
    We prove the upper bound and lower bound of the first equality separately.
    \begin{itemize}
        \item For the upper bound, consider the naive query algorithm that queries all the input bits. The instance complexity of $f$ w.r.t.~this algorithm is clearly $\frac{n}{\Cmin(f)}$.
        \item For the lower bound proof, assume without loss of generality that $\ell_0(f) \neq 0$. Indeed, if $\ell_0(f) = 0$, then we can replace $f$ by its complement $1-f$. This does not change the instance complexity and the minimum certificate complexity. The transformation can also easily be seen to preserve the quantity $\ell_0(f) + n - \ell_1(f)$. So for the rest of the proof we assume $\ell_0(f) \neq 0$, and thus $D_f(\ell_0(f) - 1) \neq D_f(\ell_0(f))$.
        
        Towards a lower bound, consider a decision tree $T$ that computes $f$. Consider the path of $T$ that answers $0$ to the first $n - \ell_0(f)$ variables, and $1$ to the all of the $\ell_0(f)$ variables after that. There must exist such a path for the following reason. If the tree terminates before this, there exist inputs reached on this path that output different answers (there are inputs consistent with the path so far that have Hamming weights $\ell_0(f)-1$ and $\ell_0(f)$). Thus, $T$ cannot compute $f$ in this case, contradicting our assumption. The Hamming weight of an input that reaches this leaf is $\ell_0(f)$. The certificate complexity of such an input is $\ell_0(f) + n - \ell_1(f)$: a certificate is a set of $\ell_0(f)$ 1's and $n - \ell_1(f)$ 0's. This concludes the proof of the lower bound.\hfill\qed
\end{itemize}
\let\qed\relax
\end{proof}

As a corollary we obtain a complete characterization of symmetric Boolean functions that are strictly $D$-instance optimizable (that is, functions with instance complexity 1). We view this as an analogous result to Grossman et al.'s characterization of monotone functions that are strictly $D$-instance optimizable.
\begin{corollary}
    The only symmetric Boolean functions that are strictly $D$-instance optimizable are the parity function $\XOR_n$ and its negation.
\end{corollary}
\begin{proof}
    It follows from Corollary~\ref{cor:symmmaxcmin} and Theorem~\ref{thm:symmic}.
\end{proof}

\section{Instance complexity of some graph properties}\label{sec:graph}
In this section we give tight bounds on the instance complexity of the graph properties of Connectivity and $k$-Clique. In the setting of graph properties, our input is a string in $\binom{n}{2}$, one variable per edge. A variable being set to 1 means the corresponding edge is present, and a variable being set to 0 means the corresponding edge is absent.
Thus, we identify an unweighted simple graph $G$ with its corresponding $\binom{n}{2}$-bit string. We now list the problems of interest to us.
\begin{definition}[Connectivity]\label{defn:conn}
    For a positive integer $n > 0$, define the function $\CONN : \binom{n}{2} \to \zone$ as $\CONN(G) = 1$ iff $G$ is connected.
\end{definition}
\begin{definition}[$k$-Clique]\label{defn:kcl}
    For positive integers $0 < k \leq n$, define the function $\kCL : \binom{n}{2} \to \zone$ as $\kCL(G) = 1$ iff $G$ contains a $k$-clique as a subgraph.
\end{definition}

Our main theorem of this section is as follows.
\begin{theorem}[Restatement of Theorem~\ref{thm:graphintro}]\label{thm:graph}
    Let $n$ be a sufficiently large positive integer and $k$ be a positive integer satisfying $k^3 \leq n^2/4$. Then,
    \[
    \IC(\CONN) = \frac{\DT(\CONN)}{\Cmin(\CONN)} = \frac{\binom{n}{2}}{n-1}, \qquad \IC(\kCL) = \frac{\DT(\kCL)}{\Cmin(\kCL)} = \frac{\binom{n}{2}}{\binom{k}{2}}.
    \]
\end{theorem}
\begin{proof}
    We first note that both of the graph properties have maximal query complexity, then analyze their $\Cmin$ values, and finally show the required bounds on their $\IC$ values.
    \begin{itemize}
        \item We first note that $\CONN$ and $\kCL$ are known to be \emph{evasive} graph properties, that is, their query complexity is $\binom{n}{2}$ (see~\cite{LY02} and~\cite{Bol76}). Thus,
        \[
        \DT(\CONN) = \DT(\kCL) = \binom{n}{2}.
        \]
        \item Next, the certificate complexity of $1$-inputs to $\CONN$ equals $n-1$: the total number of connected components in the graph must be 1 after querying all certificate variables, implying that the certificate size is at least $n-1$. In the other direction, a spanning tree with $n-1$ edges serves as a certificate for connectivity. It is easy to see that a certificate for a $0$-input cannot have less than $n-1$ edges since any such certificate must completely contain a cut, and the smallest possible cut is that defined by a single vertex of size $n-1$.
        For all values of $k$, the certificate complexity of $1$-inputs to $\kCL$ equals $\binom{k}{2}$ since any certificate must contain a $k$-clique (otherwise setting all variables outside the certificate gives a 0-input), and a single $k$-clique serves as a certificate.
        On the other hand, in order to certify $0$-inputs, we may assume that a certificate only queries 0-variables (otherwise simply drop 1-variables queried to obtain a smaller certificate). Moreover, the certificate must have the property that even if all variables outside it are set to $1$, then the graph does not contain a $k$-clique. Tur{\'a}n's theorem~\cite{Tur41} (also see~\cite{Aig95, Bol04}) states that any graph that does not contain a $k$-clique must contain at most $\frac{n^2(k-2)}{2(k-1)}$ edges. This implies that a certificate for $0$-inputs must contain at least $\binom{n}{2} - \frac{n^2(k-2)}{2(k-1)}$ variables (otherwise set all variables outside the certificate to 1, which yields a graph containing a $k$-clique by Tur{\'a}n's theorem). Moreover, there exists a graph achieving this bound~\cite{Tur41}. Thus,
        \[
        \Cmin(\CONN) = n-1, \qquad \Cmin(\kCL) = \min\cbra{\binom{k}{2}, \binom{n}{2} - \frac{n^2(k-2)}{2(k-1)}}.
        \]
By our assumption, we have $k^3 \leq n^2/4$.
        Thus,
        \begin{align*}
            \binom{k}{2} & = \frac{k(k-1)}{2} \leq \frac{k^2}{2} \leq \frac{n^2}{8k} \leq \frac{n^2}{8(k-1)} \leq \frac{n(n-k+1)}{2(k-1)} = \binom{n}{2} - \frac{n^2(k-2)}{2(k-1)}.
        \end{align*}
Here, the last inequality holds because $n/4 \leq n - k + 1 \iff k \leq 3n/4 + 1$, which is true for sufficiently large $n$ as our assumption guarantees $k \leq n^{2/3}/4^{1/3}$ , and the last equality follows from straightforward algebra.
Hence $\Cmin(\kCL) = \binom{k}{2}$ in this regime.
        \item To see the claimed bounds on instance complexity, first recall that $\IC(f) \leq \DT(f)/\Cmin(f)$ holds for all Boolean functions $f$. Consider a query algorithm solving $\CONN$. Since $\CONN$ is evasive, there exists an input such that the function value remains undetermined even after $\binom{n}{2} - 1$ queries. If the unqueried edge is $e$, this means that the graph with $e$ absent (and the other edges consistent with the queries so far) is not connected, and the graph with $e$ present is connected. The graph with $e$ present has a 1-certificate of size $n-1$, and hence 
        \[
        \IC(\CONN) \geq \frac{\binom{n}{2}}{n-1}.
        \]
        Consider a query algorithm solving $\kCL$. Again, since $\kCL$ is evasive, this implies existence of an input whose function value is undetermined before the last query. Just as in the argument for $\CONN$, let the unqueried edge be denoted by $e$. The graph with $e$ absent (and the other edges consistent with the queries so far) does not contain a $k$-clique, and the graph with $e$ present contains a $k$-clique. The graph with $e$ present has a 1-certificate of size $\binom{k}{2}$, and hence
        \[
        \IC(\CONN) \geq \frac{\binom{n}{2}}{\binom{k}{2}}.
        \]
        \end{itemize}
\end{proof}

\section{Instance complexity of some linear threshold functions}
In Section~\ref{sec:symm} we characterized the instance complexity of all symmetric Boolean functions. In particular, Theorem~\ref{thm:symmic} shows that $\IC(f) = \DT(f)/\Cmin(f)$ for all symmetric Boolean functions $f$. We also showed this bound to hold true for specific graph properties is Section~\ref{sec:graph}. This raises the natural question of whether $\IC(f) = \DT(f)/\Cmin(f)$ holds true for \emph{all} Boolean $f$. We show in this section that this is not the case in a very strong sense, and exhibit two examples witnessing this. The two examples we exhibit belong to the class of \emph{linear threshold functions}: a Boolean function $f : \zone^n \to \zone$ is said to be a linear threshold function if there exist $w_0, w_1, \dots, w_n \in \mathbb{R}$ such that $f(x) = \mathrm{sgn}(w_0 + \sum_{i = 1}^n w_ix_i)$. Here, $\mathrm{sgn}(\cdot)$ is defined to the the function that outputs 0 on a negative input, 1 on a positive input, and is undefined when the input equals 0.

The first example is the Greater-Than function that takes two $n$-bit strings as input and outputs 1 iff the first string is lexicographically larger than the second one.
\begin{definition}[Greater-Than]\label{defn:gt}
    For a positive integer $n$, the Greater-Than function on $2n$ inputs, denoted $\GT_n$, is defined by $\GT_n(x_1, \dots, x_n, y_1, \dots, y_n) = 1$ iff the integer whose binary representation is $x_n \dots, x_1$ is strictly larger than the integer whose binary representation is $y_n, \dots, y_1$. 
\end{definition}
The second example is the Odd-Max-Bit function that takes an $n$-bit string as input and outputs 1 iff the right-most variable with value 1 has an odd index.
\begin{definition}[Odd-Max-Bit]\label{defn:omb}
    For a positive integer $n$, the Odd-Max-Bit function on $n$ inputs, denoted $\OMB_n$, is defined by
    \[
    \OMB_n(x_1, \dots, x_n) = 1 \iff \max\cbra{i \in [n] : x_i = 1}~\textnormal{is odd}.
    \]
Define $\OMB(0^n) = 0$.
\end{definition}

To see why $\GT_n$ and $\OMB_n$ are linear threshold functions, observe that
\begin{align*}
    \GT_n(x_1, \dots, x_n, y_1, \dots, y_n) & = \mathrm{sgn}\left(\sum_{i = 1}^n (2^{i-1}x_i - 2^{i-1}y_i) - 0.5\right),\\
    \OMB_n(x_1, \dots, x_n) & = \mathrm{sgn}\left(\sum_{i = 1}^n (-2)^{i+1}x_i - 0.5\right).
\end{align*}

\begin{theorem}[Restatement of Theorem~\ref{thm:gtombresultsintro}]\label{thm:gtombresults}
    For all positive integers $n$, we have
    \begin{align*}
        \DT(\GT_n) & = 2n, \qquad \Cmin(\GT_n) = 2, \qquad \IC(\GT_n) \leq 2\\
        \DT(\OMB_n) & = n, \qquad \Cmin(\OMB_n) = 1, \qquad \IC(\OMB_n) \leq 2.
    \end{align*}
\end{theorem}
Before we prove the theorem, we require the following properties of Boolean functions. Every Boolean function $f : \zone^n \to \zone$ has a unique multilinear polynomial expansion of the form $f(x) = \sum_{S \subseteq [n]}\widetilde{f}(S) \prod_{i \in S}x_i$ where each $\widetilde{f}(S)$ is a real number. This expansion is sometimes referred to as the \emph{M\"obius expansion} of $f$. The \emph{degree} of $f$, denoted by $\deg(f)$, is the maximum degree of a monomial in its M\"obius expansion that has a non-zero coefficient. It is not hard to show that a depth-$d$ decision tree for $f$ induces a degree-$d$ polynomial computing $f$: sum up the indicator polynomials of each 1-leaf, and each of these indicator polynomials can easily be seen to have degree at most the depth of their corresponding leaf. This yields the following folklore lemma.
\begin{lemma}[Folklore]\label{lem:degdt}
    Let $f : \zone^n \to \zone$ be a Boolean function. Then $\DT(f) \geq \deg(f)$.
\end{lemma}
We now prove Theorem~\ref{thm:gtombresults}.

\begin{proof}[of Theorem~\ref{thm:gtombresults}]
    We first show that $\DT(\GT_n) = 2n$ and $\DT(\OMB_n) = n$, after which we show the $\Cmin$ bounds and finally we show the $\IC$ bounds. 
    \begin{itemize}
        \item Towards a $2n$ lower bound for $\DT(\GT_n)$, consider an adversary who uses the following strategy against a query algorithm for $\GT_n$ for the first $2n-1$ queries:
        \begin{itemize}
            \item If a variable $x_i$ is queried for the first time out of the pair $\cbra{x_i, y_i}$, answer 1.
            \item If a variable $y_i$ is queried for the first time out of the pair $\cbra{x_i, y_i}$, answer 0.
            \item If a variable $x_i$ ($y_i$) is queried such that $y_i$ ($x_i$) has already been queried, then answer such that $x_i = y_i$.
        \end{itemize}
        Just before the last query of the algorithm, we know that the strings are equal except for one bit. Say $x_i$ is the last unqueried bit. By the definition of the adversary, we know that $y_i = 0$. Thus, we have $\GT_n(x, y) = 0$ if $x_i = 0$ and $\GT_n(x, y) = 1$ if $x_i = 1$. That is, the function value depends on the last unqueried input bit. A similar argument works if the last unqueried bit is $y_i$. Thus, $\DT(\GT_n) \geq 2n$ and hence $\DT(\GT_n) = 2n$.
        \item The Odd-Max-Bit function outputs the parity of the largest index with variable value 1 if it exists, and outputs 0 otherwise. One may find such an index by scanning the input from right to left. This intuition is captured in the unique polynomial representation of $\OMB_n$ as $\OMB_n(x) = $
        \begin{align}\label{eqn:OMBexpansion}
        x_n\cdot 0 + (1-x_n)x_{n-1} \cdot 1 + (1-x_n)(1-x_{n-1})\OMB_{n-2}(x_1, \dots, x_{n-2}) & \quad \text{if $n$ is even, or}\\
        x_n\cdot 1 + (1-x_n)x_{n-1} \cdot 0 + (1-x_n)(1-x_{n-1})\OMB_{n-2}(x_1, \dots, x_{n-2}) & \quad \text{if $n$ is odd},
        \end{align}
        with $\OMB_1(x_1) = x_1$ and $\OMB_2(x_1, x_2) = 0 \cdot x_2 + (1-x_2)x_1 = x_1 - x_1x_2$.
        Note that in either case, the coefficient of the maximum-degree monomial equals $(-1)^{n + 1} \neq 0$, and hence the degree is $n$.
        By Lemma~\ref{lem:degdt}, this implies $\DT(\OMB_n) \geq n$, and hence $\DT(\OMB_n) = n$.
        \item The inequalities $\Cmin(\GT_n) \leq 2$ and $\Cmin(\OMB_n) \leq 1$ are witnessed by the strings $0^{n-1}10^n$ and $0^{n-1}1$, respectively. It is easy to see that no inputs $z \in \zone^{2n}, w \in \zone^n$ satisfy $\C(\GT_n, z) = 1$ and $\C(\OMB_n, w) = 0$, yielding $\Cmin(\GT_n) = 2$ and $\Cmin(\OMB_n) = 1$.
        \item For the instance complexity of Greater-Than, consider the natural query algorithm $\cA$ that first queries the most significant bits of $x, y$, outputs the answer if we can deduce it at this point, and use an algorithm for the Greater-Than instance without these two bits if we cannot deduce the output at this point, i.e., the bits seen are equal (see Figure~\ref{fig:gtdecisiontree} for a visual description of the query algorithm $\cA$ used).

    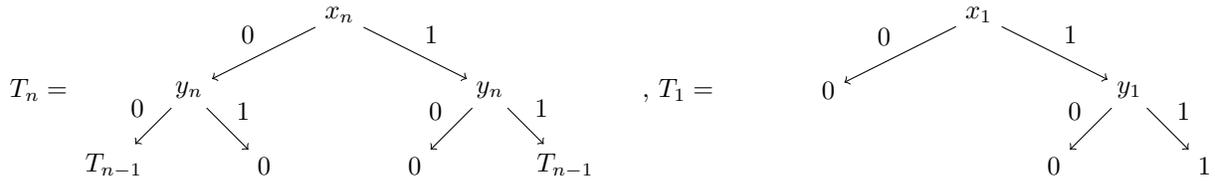
\begin{figure}[h!]
	\centering
	\begin{tabular}{ccc}
		\begin{tikzpicture}[yscale=-1, xscale=.7]
  \node (tn) at (-4,1) {$T_n = $};
			\node (root) at (0,0) {$x_n$};
			\node (1) at (2,1) {$y_n$};
			\node (11) at (3,2) {$T_{n-1}$};
			\node (0) at (-2,1) {$y_n$};
			\node (10) at (1,2) {$1$};
			\node (00) at (-3,2) {$T_{n-1}$};
			\node (01) at (-1,2) {$0$};
			\draw[->] (root) to node[above right] {$1$} (1);
			\draw[->] (0) to node[above right] {$1$} (01);
			\draw[->] (root) to node[above left] {$0$} (0);
			\draw[->] (1) to node[above left] {$0$} (10);
			\draw[->] (0) to node[above left] {$0$} (00);
			\draw[->] (1) to node[above right] {$1$} (11);

   \node (t1) at (4.5, 1) {, $T_1 = $};
			\node (sroot) at (8.5,0) {$x_1$};
			\node (s1) at (10.5,1) {$y_1$};
			\node (s11) at (11.5,2) {0};
			\node (s0) at (6.5,1) {$0$};
			\node (s10) at (9.5,2) {1};
			\draw[->] (sroot) to node[above right] {$1$} (s1);
			\draw[->] (sroot) to node[above left] {$0$} (s0);
			\draw[->] (s1) to node[above left] {$0$} (s10);
			\draw[->] (s1) to node[above right] {$1$} (s11);
\end{tikzpicture}
  \end{tabular}
	\caption{A decision tree $T_n$ for Greater-Than on $2n$ input bits}
	\label{fig:gtdecisiontree}
\end{figure}
We now argue that this algorithm $\cA$ witnesses $\IC(\GT_n) \leq 2$. We analyze the instance complexity of each input with respect to $\cA$.
\begin{itemize}
    \item Consider an input of the form $(x, y)$ with $x = y$. By definition, this is a 0-input for $\GT_n$. In order to certify that $x \ngtr y$, it suffices to query all of the 0-variables of $x$ and all of the 1-variables of $y$. It is also necessary that a certificate queries at least one element of each pair $(x_i, y_i)$, since otherwise the function value could be set to 1 by setting $x_i = 1$ and $y_i = 0$, a contradiction to the assumption that we started with a certificate. Thus, $\C(\GT_n, (x,y)) = n$ for all $x = y$. The number of input bits read by our query algorithm can easily be seen to be in $\cbra{2n-1, 2n}$. Thus, $\IC(\GT_n, (x, y), \cA) \leq \frac{2n}{n} = 2$ for all inputs with $x = y$.
    \item For an input $(x, y)$ with $x \neq y$, let $n - j$ denote the largest index with $x_{n-j} \neq y_{n-j}$. Here, $j \in \cbra{0, 1, \dots, n-1}$. Just as the argument in the previous bullet, it is easy to show that a certificate for $(x,y)$ must query at least one variable from each pair $(x_i, y_i)$ with $i \geq n-j$.\footnote{Moreover, there exists a certificate that makes one extra query: if $(x,y)$ is a 1-input, query all of the $x_i$ with $i > n - j$ and $x_i = 1$, and query all of the $y_i$ with $i > n - j$ and $y_i = 0$. Finally, query the pair $(x_{n-j}, y_{n-j})$.} Thus, $\C(\GT_n, (x, y)) \geq j+1$ for all such inputs. The number of queries made by our algorithm $\cA$ can be seen to be in $\cbra{2j + 1, 2j + 2}$: it queries all pairs $(x_i, y_i)$ (except when $n - j = 1$ and $x_1 = 0$). This implies $\IC(\GT_n, (x, y), \cA) \leq \frac{2j+2}{j+1} = 2$ for all inputs $(x, y)$ with $n-j$ the largest index where $x_{n-j} \neq y_{n-j}$.
\end{itemize}
\item As in the argument for the instance complexity of $\GT_n$, consider the natural query algorithm $\cB$ for Odd-Max-Bit that queries the variables from right to left and outputs the parity of the first index seen where the input takes value 1 (see Figure~\ref{fig:ombdecisiontreeeven} for a visual description of $\cB$ when $n$ is even and Figure~\ref{fig:ombdecisiontreeodd} for when $n$ is odd).
    \begin{figure}[h!]
	\centering
	\begin{tabular}{ccc}
		\begin{tikzpicture}[yscale=-1]
  \node (tn) at (-3,1) {$T_n = $};
			\node (root) at (0,0) {$x_n$};
			\node (1) at (1,1) {$0$};
			\node (0) at (-1,1) {$x_{n-1}$};
			\node (00) at (-2,2) {$T_{n-2}$};
			\node (01) at (0,2) {$1$};
			\draw[->] (root) to node[above right] {$1$} (1);
			\draw[->] (0) to node[above right] {$1$} (01);
			\draw[->] (root) to node[above left] {$0$} (0);
			\draw[->] (0) to node[above left] {$0$} (00);

   \node (t1) at (2.5, 1) {, $T_2 = $};
			\node (sroot) at (5,0) {$x_2$};
			\node (s1) at (6,1) {$0$};
			\node (s0) at (4,1) {$x_1$};
            \node (s00) at (3,2) {0};
            \node (s01) at (5,2) {1};
			\draw[->] (s0) to node[above right] {$1$} (s01);
			\draw[->] (s0) to node[above left] {$0$} (s00);
			\draw[->] (sroot) to node[above right] {$1$} (s1);
			\draw[->] (sroot) to node[above left] {$0$} (s0);
\end{tikzpicture}
  \end{tabular}
	\caption{A decision tree $T_n$ for Odd-Max-Bit on $n$ input bits with $n$ even}
	\label{fig:ombdecisiontreeeven}
\end{figure}
    \begin{figure}[h!]
	\centering
	\begin{tabular}{ccc}
		\begin{tikzpicture}[yscale=-1]
  \node (tn) at (-3,1) {$T_n = $};
			\node (root) at (0,0) {$x_n$};
			\node (1) at (1,1) {$1$};
			\node (0) at (-1,1) {$x_{n-1}$};
			\node (00) at (-2,2) {$T_{n-2}$};
			\node (01) at (0,2) {$0$};
			\draw[->] (root) to node[above right] {$1$} (1);
			\draw[->] (0) to node[above right] {$1$} (01);
			\draw[->] (root) to node[above left] {$0$} (0);
			\draw[->] (0) to node[above left] {$0$} (00);

   \node (t1) at (2.5, 1) {, $T_1 = $};
			\node (sroot) at (5,0) {$x_1$};
			\node (s1) at (6,1) {$1$};
			\node (s0) at (4,1) {$0$};
			\draw[->] (sroot) to node[above right] {$1$} (s1);
			\draw[->] (sroot) to node[above left] {$0$} (s0);
\end{tikzpicture}
  \end{tabular}
	\caption{A decision tree $T_n$ for Odd-Max-Bit on $n$ input bits with $n$ odd}
	\label{fig:ombdecisiontreeodd}
\end{figure}
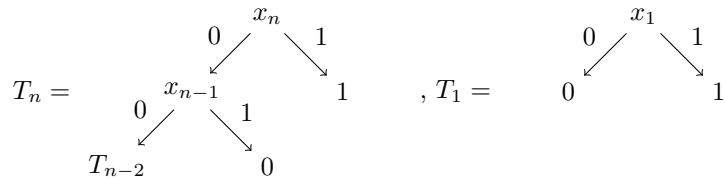
We now argue that this algorithm $\cB$ witnesses $\IC(\OMB_n) \leq 2$. We analyze the instance complexity of each input with respect to $\cB$.
\begin{itemize}
    \item Let $x$ be a 0-input to $\OMB_n$. First, consider the input $x = 0^n$. Any certificate for this input must query all variables $x_i$ with $i$ odd, since if unqueried, we could set $x_i = 1$, forcing $\OMB_n = 1$. Moreover, the set of all variables $x_i$ with $i$ odd forms a certificate for $0^n$. Thus, $\C(\OMB_n, 0^n) = \lfloor(n + 1)/2\rfloor$. The algorithm $\cB$ queries all variables on this input, and hence $\IC(\OMB_n, 0^n, \cB) = \frac{n}{\lfloor(n + 1)/2\rfloor} \leq 2$. Next, consider a 0-input $x \neq 0^n$ with $n-i$ the maximum index satisfying $x_{n-i} = 1$. Since $x$ is a 0-input, $n-i$ is even. A certificate for $x$ must query all variables $x_j$ with $j > n-i$ and $j$ odd, since otherwise setting $x_j = 1$ forces the output of $\OMB_n$ to 1. It must also query at least one more variable. Moreover there exists such a certificate, where the extra variable queried is $x_{n-i}$. Thus, $\C(\OMB_n, x) = \lfloor(i+1)/2\rfloor + 1$. The algorithm $\cB$ queries $i+1$ variables on this input, and hence $\IC(\OMB_n, x, \cB) = \frac{i+1}{\lfloor(i+1)/2\rfloor + 1} < 2$.
    \item Consider a 1-input $x$ with $n-i$ the maximum index satisfying $x_{n-i} = 1$. Since $x$ is a 1-input, $n-i$ is odd. A certificate for $x$ must query all variables $x_j$ with $j > n-i$ and $j$ even, since otherwise setting $x_j = 1$ forces the output of $\OMB_n$ to 0. It must also query at least one more variable. Moreover there exists such a certificate, where the extra variable queried is $x_{n-i}$. Thus, $\C(\OMB_n, x) = \lfloor(i+1)/2\rfloor +1$. The algorithm $\cB$ queries $i+1$ variables on this input, and hence $\IC(\OMB_n, x, \cB) = \frac{i+1}{\lfloor(i+1)/2\rfloor + 1} < 2$.
\end{itemize}
\end{itemize}
\end{proof}

\bibliographystyle{abbrvnat}
\bibliography{bibo}

\begin{thebibliography}{14}
\providecommand{\natexlab}[1]{#1}
\providecommand{\url}[1]{\texttt{#1}}
\expandafter\ifx\csname urlstyle\endcsname\relax
  \providecommand{\doi}[1]{doi: #1}\else
  \providecommand{\doi}{doi: \begingroup \urlstyle{rm}\Url}\fi

\bibitem[Aigner(1995)]{Aig95}
M.~Aigner.
\newblock Tur{\'a}n's graph theorem.
\newblock \emph{The American Mathematical Monthly}, 102\penalty0 (9):\penalty0
  808--816, 1995.

\bibitem[Ben-Eliezer et~al.(2023)Ben-Eliezer, Grossman, and Naor]{BGN23}
O.~Ben-Eliezer, T.~Grossman, and M.~Naor.
\newblock On the instance optimality of detecting collisions and subgraphs.
\newblock \emph{arXiv preprint arXiv:2312.10196}, 2023.

\bibitem[Bollob{\'a}s(1976)]{Bol76}
B.~Bollob{\'a}s.
\newblock Complete subgraphs are elusive.
\newblock \emph{Journal of Combinatorial Theory, Series B}, 21\penalty0
  (1):\penalty0 1--7, 1976.

\bibitem[Bollobas(2004)]{Bol04}
B.~Bollobas.
\newblock \emph{Extremal Graph Theory}.
\newblock Dover Publications, Inc., 2004.

\bibitem[Buhrman and de~Wolf(2002)]{BW02}
H.~Buhrman and R.~de~Wolf.
\newblock Complexity measures and decision tree complexity: a survey.
\newblock \emph{Theor. Comput. Sci.}, 288\penalty0 (1):\penalty0 21--43, 2002.
\newblock \doi{10.1016/S0304-3975(01)00144-X}.
\newblock URL \url{https://doi.org/10.1016/S0304-3975(01)00144-X}.

\bibitem[Ding et~al.(2023)Ding, Feng, Ho, Tang, and
  Xu]{DBLP:conf/soda/DingFHTX23}
B.~Ding, Y.~Feng, C.~Ho, W.~Tang, and H.~Xu.
\newblock Competitive information design for {P}andora's box.
\newblock In \emph{Proceedings of the 2023 {ACM-SIAM} Symposium on Discrete
  Algorithms, {SODA} 2023}, pages 353--381. {SIAM}, 2023.
\newblock \doi{10.1137/1.9781611977554.ch15}.
\newblock URL \url{https://doi.org/10.1137/1.9781611977554.ch15}.

\bibitem[Erlebach et~al.(2016)Erlebach, Hoffmann, and
  Kammer]{DBLP:journals/tcs/Erlebach0K16}
T.~Erlebach, M.~Hoffmann, and F.~Kammer.
\newblock Query-competitive algorithms for cheapest set problems under
  uncertainty.
\newblock \emph{Theor. Comput. Sci.}, 613:\penalty0 51--64, 2016.
\newblock \doi{10.1016/j.tcs.2015.11.025}.
\newblock URL \url{https://doi.org/10.1016/j.tcs.2015.11.025}.

\bibitem[Feder et~al.(2000)Feder, Motwani, Panigrahy, Olston, and
  Widom]{DBLP:conf/stoc/FederMPOW00}
T.~Feder, R.~Motwani, R.~Panigrahy, C.~Olston, and J.~Widom.
\newblock Computing the median with uncertainty.
\newblock In \emph{Proceedings of the Thirty-Second Annual {ACM} Symposium on
  Theory of Computing}, pages 602--607. {ACM}, 2000.
\newblock \doi{10.1145/335305.335386}.
\newblock URL \url{https://doi.org/10.1145/335305.335386}.

\bibitem[Feder et~al.(2007)Feder, Motwani, O'Callaghan, Olston, and
  Panigrahy]{DBLP:journals/jal/FederMOOP07}
T.~Feder, R.~Motwani, L.~O'Callaghan, C.~Olston, and R.~Panigrahy.
\newblock Computing shortest paths with uncertainty.
\newblock \emph{J. Algorithms}, 62\penalty0 (1):\penalty0 1--18, 2007.
\newblock \doi{10.1016/j.jalgor.2004.07.005}.
\newblock URL \url{https://doi.org/10.1016/j.jalgor.2004.07.005}.

\bibitem[Grossman et~al.(2020)Grossman, Komargodski, and Naor]{GKN20}
T.~Grossman, I.~Komargodski, and M.~Naor.
\newblock Instance complexity and unlabeled certificates in the decision tree
  model.
\newblock In \emph{11th Innovations in Theoretical Computer Science Conference,
  {ITCS}}, volume 151 of \emph{LIPIcs}, pages 56:1--56:38. Schloss Dagstuhl -
  Leibniz-Zentrum f{\"{u}}r Informatik, 2020.
\newblock \doi{10.4230/LIPIcs.ITCS.2020.56}.
\newblock URL \url{https://doi.org/10.4230/LIPIcs.ITCS.2020.56}.

\bibitem[Halld{\'{o}}rsson and de~Lima(2021)]{DBLP:journals/tcs/HalldorssonL21}
M.~M. Halld{\'{o}}rsson and M.~S. de~Lima.
\newblock Query-competitive sorting with uncertainty.
\newblock \emph{Theor. Comput. Sci.}, 867:\penalty0 50--67, 2021.
\newblock \doi{10.1016/j.tcs.2021.03.021}.
\newblock URL \url{https://doi.org/10.1016/j.tcs.2021.03.021}.

\bibitem[Hoffmann et~al.(2008)Hoffmann, Erlebach, Krizanc, Mihal{\'{a}}k, and
  Raman]{DBLP:conf/stacs/HoffmannEKMR08}
M.~Hoffmann, T.~Erlebach, D.~Krizanc, M.~Mihal{\'{a}}k, and R.~Raman.
\newblock Computing minimum spanning trees with uncertainty.
\newblock In \emph{{STACS} 2008, 25th Annual Symposium on Theoretical Aspects
  of Computer Science}, volume~1 of \emph{LIPIcs}, pages 277--288. Schloss
  Dagstuhl - Leibniz-Zentrum f{\"{u}}r Informatik, Germany, 2008.
\newblock \doi{10.4230/LIPIcs.STACS.2008.1358}.
\newblock URL \url{https://doi.org/10.4230/LIPIcs.STACS.2008.1358}.

\bibitem[Lov{\'a}sz and Young(2002)]{LY02}
L.~Lov{\'a}sz and N.~E. Young.
\newblock Lecture notes on evasiveness of graph properties.
\newblock \emph{arXiv preprint cs/0205031}, 2002.

\bibitem[Tur{\'a}n(1941)]{Tur41}
P.~Tur{\'a}n.
\newblock On an external problem in graph theory.
\newblock \emph{Mat. Fiz. Lapok}, 48:\penalty0 436--452, 1941.

\end{thebibliography}

\end{document}